\documentclass[preprint,authoryear]{elsarticle}
\usepackage{amsmath}
\usepackage{graphicx}
\usepackage{enumerate}

\usepackage{amssymb}
\usepackage{amsthm}
\graphicspath{{figures/}}
\usepackage[colorlinks=true, allcolors=blue]{hyperref}
\usepackage{enumerate} % For roman numerals in enumerate
\usepackage{soul}
\usepackage{placeins} % FloatBarrier

\newcommand{\pr}[1]{\Pr[#1]}
\newcommand{\simiid}{\overset{i.i.d.}{\sim}}

\newtheorem{lemma}{Lemma}

\newtheorem{corollary}{Corollary}
\newtheorem{proposition}{Proposition}

\theoremstyle{definition}
\newtheorem{definition}{Definition}

\theoremstyle{remark}
\newtheorem{remark}{Remark}
\begin{document}

\def\spacingset#1{\renewcommand{\baselinestretch}%
{#1}\small\normalsize} \spacingset{1}

%%%%%%%%%%%%%%%%%%%%%%%%%%%%%%%%%%%%%%%%%%%%%%%%%%%%%%%%%%%%%%%%%%%%%%%%%%%%%%
\title{Fast calculation of p-values for one-sided Kolmogorov-Smirnov type statistics}
\author[1]{Amit Moscovich}%\corref{cor1}}
\fntext[fn1]{Department of Statistics and Operations Research,  Tel Aviv University, Tel Aviv, 69978,
Israel.}
\ead{mosco@tauex.tau.ac.il}

\begin{abstract}
    A novel method for computing exact p-values of one-sided statistics from the Kolmogorov-Smirnov family is presented.
    It covers the Higher Criticism statistic, one-sided weighted Kolmogorov-Smirnov statistics, and the one-sided Berk-Jones
statistics.
    In addition to p-values, the method can also be used for power analysis, finding alpha-level thresholds, and the construction
of confidence bands for the empirical distribution function.

    With its quadratic runtime and numerical stability, the method easily scales to sample sizes in the hundreds of thousands
and takes less than a second to run on a sample size of 25,000. This allows practitioners working on large data sets to use
exact finite-sample computations instead of approximation schemes.

The method is based on a reduction to the boundary-crossing probability of a pure jump stochastic process. FFT convolutions
of two different sizes are then used to efficiently propagate the probabilities of the non-crossing paths.
This approach has applications beyond statistics, for example in financial risk modeling.
\end{abstract}

\begin{keyword}
Continuous goodness-of-fit \sep
Higher criticism \sep
Stochastic process \sep
Boundary crossing \sep
Hypothesis testing
\end{keyword}

\maketitle

\spacingset{1.5} % DON'T change the spacing!

%%%%%%%%%%%%%%%%%%%%%%%%%%%%%%%%%%
\section{Introduction}
%%%%%%%%%%%%%%%%%%%%%%%%%%%%%%%%%%

Let $X_1, \ldots, X_n$ be random variables drawn independently from a distribution~$F$ and let $X_{(1)} \le X_{(2)} \le \cdots \le X_{(n)}$ be their order statistics.
In this paper, we present a fast and numerically stable algorithm for computing  one-sided non-crossing probabilities  of the form
\begin{align} \label{eq:upperbound_problem}
    \pr{\forall i: X_{(i)} \le \beta_i},
\end{align}
where $\beta_1, \ldots, \beta_n$ are (arbitrary) upper bounds.
This probability may be rewritten as
\begin{align} \label{eq:order_statistics_upperbound_problem}
    \text{NCPROB}(B_1, \ldots, B_n)
    &:=
    \pr{\forall i: U_{(i)} \le B_i \big| U_1, \ldots, U_n \simiid U[0,1]}.
\end{align}
where $B_i=F(\beta_i)$ and $U_{(1)} \le \cdots \le U_{(n)}$ are the order statistics of a uniform sample in~$[0,1]$.
This equivalence follows by  expressing $X_i$ using the inverse transformation $X_i = F^{-1}(U_i)$  where $F^{-1}(u) := \inf \{x \in \mathbb{R}: F(x)
\geq u \}$ is the generalized inverse distribution function and then noting that  $F^{-1}(U_{(i)}) \le \beta_i$  holds if and only if $U_{(i)} \le F(\beta_i)=B_i$.

A closely related problem, that is also covered by our algorithm, is the computation of one-sided non-crossing probabilities for the empirical cumulative distribution
function (eCDF).
Given a function $b:[0,1] \to \mathbb{R}$,
this is the probability that $b(t)$ bounds the empirical CDF from below,
\begin{align} \label{eq:pr_binomial_no_cross}
    \pr{\forall t \in [0,1]: b(t) \le nF_n(t) \big| U_1, \ldots, U_n \simiid U[0,1]}.
\end{align}
where $F_n(t) := \tfrac1n \sum_{i=1}^n \mathbf{1}(U_i \le t)$ is the eCDF of the sample $U_1, \ldots, U_n$.
The non-crossing probability~\eqref{eq:pr_binomial_no_cross} is equal to the non-crossing probability
$\text{NCPROB}(B_1, \ldots, B_n)$  with upper bounds given by the first integer crossings of the lower boundary function \citep{Gleser1985},
\begin{align} \label{def:B_i}
    B_{i} = \inf\{t \in [0,1] : b(t) > i-1 \}.
\end{align}
Hence, methods for computing the probability \eqref{eq:order_statistics_upperbound_problem} can be readily applied to the calculation of non-crossing probabilities for the empirical CDF.
See Figure \ref{fig:one_sided_crossing_illustration} for an illustration.
Conversely, given a set of upper bounds $B_1, \ldots, B_n$, one may construct a step function $b(t) = \sum_{i=1}^n \mathbf{1}(B_i \le t)$ for which  the non-crossing probability~\eqref{eq:pr_binomial_no_cross} is equal to NCPROB$(B_1, \ldots, B_n)$.
To conclude, the calculation of the probabilities  \eqref{eq:order_statistics_upperbound_problem} and \eqref{eq:pr_binomial_no_cross} are two different formulations of the same problem.
This equivalence is well-known in the literature and has also been extended to discontinuous distributions \citep{Steck1971,Gleser1985,Dimitrova2020a}.
Since it is fundamental to our algorithm description, we include a concise proof of this equivalence in \ref{appendix:reduction}.

\subsection{Outline}

The main contribution of this paper is a fast and numerically stable $O(n^2)$ algorithm for computing the 
\emph{one-sided non-crossing probabilities}~\eqref{eq:order_statistics_upperbound_problem} and \eqref{eq:pr_binomial_no_cross}.
In Section~\ref{sec:motivation} we describe the application of our method to Kolmogorov-Smirnov-type goodness of fit testing and list several other potential applications.
In Section \ref{sec:existing} we review the existing methods for computing one-sided and two-sided non-crossing probabilities.
In sections \ref{sec:background} and \ref{sec:method} we describe the proposed algorithm in detail. In Section~\ref{sec:benchmarks}
we apply our method to the computation of  $p$-values for a one-sided statistic by \cite{BerkJones1979} with sample sizes up to one million, demonstrating state-of-the-art performance.
The full source code is linked in Section~\ref{sec:code}.

\begin{figure}[t]
    \qquad\ \  
    \includegraphics[width=0.9\linewidth]{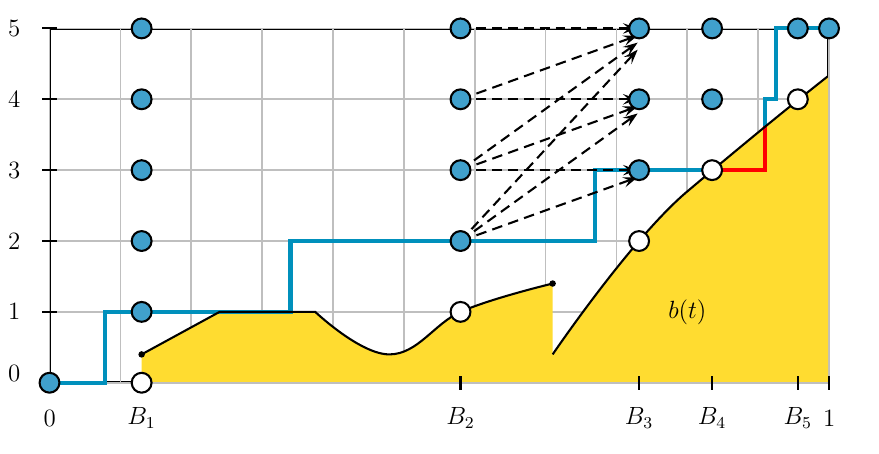}
    \caption{
       A one-sided lower boundary function $b(t)$ (in yellow) and a non-decreasing step function $f:[0,1] \to \{0, 1, \ldots\}$ (in blue) with increments at $U_1, \ldots, U_5 \in [0,1]$.
       The empty circles mark the first integer crossings of \(b(t)\) as defined in~\eqref{def:B_i}.
       By Lemma~\ref{lemma:reduction_to_finite} (\ref{appendix:reduction}),  $f(t)$ does not cross $b(t)$ if and only if for all $i$, $f(B_i) \ge i$, or equivalently, that the order statistics $U_{(1)} \le \cdots \le U_{(5)}$ all satisfy $U_{(i)} \le B_i$. 
       The blue circles  thus define a layer graph of possible non-crossing transitions for $f(t)$. In this example, $f(t)$ crosses $b(t)$ at $B_4$.
   }
   \label{fig:one_sided_crossing_illustration}
\end{figure}

\section{Motivation} \label{sec:motivation}
The primary motivation for this work is the computation of $p$-values and power for a large family of one-sided continuous
goodness-of-fit statistics.
Examples include the Higher-Criticism statistic \citep{DonohoJin2004}, one-sided variants of the Kolmogorov-Smirnov statistic
\citep{Kolmogorov1933,Renyi1953,Eicker1979,Jaeschke1979,MasonSchuenemeyer1983,JagerWellner2004}, variants of the one-sided
Berk-Jones statistics \citep{BerkJones1979,JagerWellner2005}, $\phi$-divergence statistics \citep{JagerWellner2007}, tests
based on local-levels \citep{FinnerGontscharuk2018}, and gGOF statistics \citep{ZhangJinWu2020}.
All of these one-sided statistics have the maximum form (or an equivalent minimum form)
\begin{align} \label{eq:generalized_gof}
    S := \max_{i=1,\ldots,n} s_i(F(x_{(i)})),
\end{align}
where $x_{(1)} \le \ldots \le x_{(n)}$ are the order statistics of a sample that, under the null hypothesis, is drawn from
a continuous distribution $F$, and $s_1, \ldots, s_n:\mathbb{R} \to \mathbb{R}$ are either all monotone increasing functions or all monotone decreasing functions.
For example, the one-sided Kolmogorov-Smirnov statistics are
\begin{align}
    D_n^+ := \max_{i=1, \ldots, n} \left( \frac{i}{n} - F(x_{(i)}) \right),
    \qquad
    D_n^- := \max_{i=1, \ldots, n} \left( F(x_{(i)}) - \frac{i-1}{n} \right).
\end{align}
Here, $D_n^+$ is a maximum over the monotone decreasing functions $s_i(u) = \tfrac{i}{n}-u$ and  $D_n^-$ is a maximum over monotone increasing functions $s_i(u)=u-\tfrac{i-1}{n}$.

Another example for a statistic of the form \eqref{eq:generalized_gof} is the Higher Criticism statistic of \cite{DonohoJin2004},
\begin{align}
    \text{HC}_n^* := \sqrt{n} \max_{1 \le i \le \alpha_0 \cdot n} \frac{\tfrac{i}{n} - F(x_{(i)})}{\sqrt{F(x_{(i)})(1-F(x_{(i)}))}}.
\end{align}
The HC$_n^*$ statistic can be viewed as a variant of the one-sided Kolmogorov-Smirnov statistic which takes
the maximum standardized deviation of the transformed order statistics $F(x_{(1)}), \ldots, F(x_{(n)})$ from their respective expectations.

Rather than maximizing over  standardized deviations, or Z-scores, of the transformed order statistics, one can instead consider the one-sided $p$-value of $F(x_{(i)})$ with respect to the null distribution of uniform order statistics $F(x_{(i)}) \sim \text{Beta}(i,n-i+1)$,  and take the minimum over all such $p$-values.
This is the one-sided $M_n^+$ statistic of \cite{BerkJones1979},
which has the minimum form (analogous to Eq. \eqref{eq:generalized_gof}),
\begin{align} \label{def:M_n_plus}
        M_n^+ := \min_{i=1,\ldots,n} s_i(F(x_{(i)})),
\end{align}
 where $s_i$ are the (monotone-increasing) CDFs of the corresponding Beta distributions,
 \begin{align}
    s_i(u)
    &:=
    \pr{\text{Beta}(i,n-i+1) < u}
    =
    \frac{n!}{(i-1)!(n-i)!}\int_0^u  t^{i-1} (1-t)^{n-i} dt.
\end{align}
In Section~\ref{sec:benchmarks} we present an application of our method for computing one-sided non-crossing probabilities to the computation of $p$-values for the $M_n^+$ statistic.

In the next subsections, we describe in detail how the computation of $p$-values and power for one-sided statistics of the
general form \eqref{eq:generalized_gof}  can be reduced to a calculation of the probability~\eqref{eq:order_statistics_upperbound_problem},
we discuss test statistic distribution inversion for obtaining $\alpha$-level thresholds and mention
 some applications that involve the non-crossing probability \eqref{eq:order_statistics_upperbound_problem}.

\begin{remark}
    An alternative to exact computation is the use of asymptotics.
    For the Higher Criticism, Berk-Jones, and some related statistics, the asymptotic distributions are known~\citep{Eicker1979,
    Jaeschke1979,WellnerKoltchinskii2003,MoscovichNadlerSpiegelman2016}.
    Unfortunately, the convergence of the null distribution to its limiting form can be exceedingly slow \citep{GontscharukLandwehrFinner2015},
    rendering the asymptotics inapplicable.
    More sophisticated approximations were developed (for example, by \cite{LiSiegmund2015}), but these are specific to a     particular statistic and the quality of their approximation is difficult to analyze.
    Exact finite-sample computations are generally preferable, provided that they are fast enough to be practical.
\end{remark}

\subsection[p-value and power calculations]{$p$-value and power calculations} \label{sec:pvalue}
Assume that a sample $x_1, \ldots, x_n$ is drawn independently from a continuous distribution $F$ and let $S$ be a statistic of the maximum form \eqref{eq:generalized_gof}.
Clearly, $S \le s$ if and only if $s_i(F(x_{(i)})) \le s$ for all $i$.
Since $s_i$ is monotone increasing, this occurs if and only if $F(x_{(i)}) \le s_i^{-1}(s)$.
The distribution of the statistic $S$ under the null hypothesis, that $X_i \simiid F$ is thus
\begin{align} 
    \pr{S \le s \big| X_i \simiid F}
    &=
    \pr{\forall i: F(X_{(i)}) \le s_i^{-1}(s) \big| X_i \simiid F}. \label{eq:pvalue}
\end{align}
Let $U_i = F(X_i)$, since  $F$ is continuous, we have that $U_i \sim U[0,1]$ and  $U_{(i)} = F(X_{(i)})$.
This means that \eqref{eq:pvalue} can be rewritten as
\begin{align}
    \pr{\forall i: U_{(i)} \le s_i^{-1}(s) \big| U_i \simiid U[0,1]}
    .
\end{align}
Thus the computation of the distribution of a maximum statistic $S$ as defined by Equation~\eqref{eq:generalized_gof} reduces to the calculation of the probability~\eqref{eq:order_statistics_upperbound_problem}
with $B_i=s_i^{-1}(s)$.
The $p$-value of the statistic $S$ is  given by
\begin{align}
    \text{p-value}(s) = \pr{S \ge s \big| X_i \simiid F} = 1-\text{NCPROB}(s_1^{-1}(s), \ldots, s_n^{-1}(s)).
\end{align}
Computing the power of such a statistic against a known alternative that $X_1, \ldots,
X_n \simiid G$ similarly reduces to Eq. \eqref{eq:order_statistics_upperbound_problem}, since in that case
\begin{align}
    \pr{S \ge s \big| X_i \simiid G}
    &=
    1-\pr{\forall i: s_i(F(X_{(i)})) < s \big| X_i \simiid G} \\
    &=
    1-\pr{\forall i: G(X_{(i)}) \le G(F^{-1}(s_i^{-1}(s))) \big| X_i \simiid G} \\
    &=
    1-\pr{U_{(i)} \le G(F^{-1}(s_i^{-1}(s))) \big| U_i \simiid U[0,1] } \\
    &=
    1 - \text{NCPROB}\Big(G\big(F^{-1}(s_1^{-1}(s))\big) , \ldots, G\big(F^{-1}(s_n^{-1}(s))\big)\Big).
\end{align}
For a more intricate analysis that considers distributions with discontinuities, see the analyses of \cite{Gleser1985,Dimitrova2020a}.

\subsection[Computation of alpha-level thresholds]{Computation of $\alpha$-level thresholds} \label{sec:threshold}
Given a test statistic $S$ of the maximum form in Eq.~\eqref{eq:generalized_gof}, how can we pick a threshold to obtain an $\alpha$-level test?
This is the threshold $s_{n,\alpha}$ that satisfies $\pr{S \ge s_{n,\alpha} \big| X_i \simiid F} = \alpha$.
Since the probability $\pr{S \ge s}$ is monotone-decreasing in $s$, a common approach is to  find $s_{n,\alpha}$  by repeated bisection,
thus inverting the cumulative distribution of the statistic $S$ numerically.
If  we know that $s_{n,\alpha} \in [a,b]$ then an approximation of $s_{n,\alpha}$ with additive error $<\epsilon$ may be
obtained using binary search.
This search involves $O(\log((b-a)/\epsilon ))$ calculations of probabilities of the form \eqref{eq:order_statistics_upperbound_problem}.
When the range of $s_{n,\alpha}$ is not known in advance, one can use a doubling search \citep{BentleyYao1976} to obtain an $\epsilon$-approximation
with $O(\log(s_{n,\alpha}/\epsilon))$ probability calculations of the form \eqref{eq:order_statistics_upperbound_problem}.
\subsection{Additional applications}
Additional applications which involve probabilities of the form \eqref{eq:order_statistics_upperbound_problem} and \eqref{eq:pr_binomial_no_cross}
 include the construction of confidence bands for empirical distribution functions and Q-Q plots \citep{Owen1995, Frey2008, Matthews2013, Weine2023},
 multiple hypothesis testing \citep{MeinshausenRice2006,RoquainVillers2011,SchroederDickhaus2020,MiecznikowskiWang2023},
change-point detection \citep{Worsley1986},
sequential testing \citep{Dongchu1998},
financial risk modeling \citep{DimitrovaIgnatovKaishev2017,Goffard2019},
genome-wide association studies \cite{SabattiEtal2009,BarnettMukherjeeLin2017,SunLin2019,Liu2022},
exoplanet detection~\citep{SulisMaryBigot2017},
cryptography~\citep{DingEtal2018},
econometrics~\citep{GoldmanKaplan2018},
and inventory management \citep{Dimitrova2020}.

\section{Existing methods} \label{sec:existing}

In this section, we review leading computational methods for evaluating non-crossing probabilities.
We begin with methods for computing one-sided non-crossing probabilities of the form~\eqref{eq:order_statistics_upperbound_problem}
 and then proceed to two-sided non-crossing probabilities of the form
\begin{align} \label{eq:two_sided}
    \pr{\forall i: b_i \le U_{(i)} \le B_i \big| U_1, \ldots, U_n \simiid U[0,1]}.
\end{align} 
Note that any algorithm for computing two-sided non-crossing probabilities is, in particular, applicable to the one-sided problem~\eqref{eq:order_statistics_upperbound_problem} by setting $b_i = 0$ for all $i$.

\subsection{One-sided boundaries}
Many methods for computing or estimating the one-sided non-crossing probability \eqref{eq:order_statistics_upperbound_problem} have
been proposed over the years.
One approach is to  repeatedly generate $X_1, \ldots, X_n \simiid F$
and measure the percentage of times that the inequalities $X_{(i)} \le \beta_i$ hold.
This Monte-Carlo approach does not yield accurate results and can be slow when the probability of interest is small and the
sample size \(n\) is large.

For the exact computation of the non-crossing probability \eqref{eq:order_statistics_upperbound_problem}, first note that
for each set of order statistics with no repetitions $X_{(1)} < X_{(2)} < \cdots < X_{(n)}$ there are exactly $n!$ instances
of  $(X_1, \ldots, X_n)$ that map to it. Let $U_i = F(X_i)$. For a continuous $F$ we have $U_i \simiid U[0,1]$, the density of the random vector $(U_1, \ldots, U_n)$ is equal to 1 on the unit cube. It follows that
the density of the sorted vector of order statistics $(U_{(1)}, \ldots, U_{(n)})$ is equal to $n!$ on the simplex that satisfies
$U_{(1)} < \cdots < U_{(n)}$ and zero elsewhere (we may ignore events of measure zero that $U_{(i)} = U_{(i+1)}$).
It follows that
\begin{align}
    \pr{\forall i: U_{(i)} \le B_i}
    &=
    \pr{\forall i: U_{(i)} < B_i} \\
    &=
    n! \pr{\forall i: U_i < B_i  \text{ and }  U_1 < U_2 < \cdots < U_n}.
\end{align}
This may be decomposed recursively as
\begin{align}
    &n! \pr{\forall i: U_i < B_i  \text{ and }  U_1 < U_2 < \cdots < U_n} \\
    &=
    n! \int_0^{B_1} d U_1 \pr{\forall i=2, \ldots, n: U_i < B_i \text{ and } U_1 < U_2 < \cdots < U_n | U_1}  \\
    &=
    n! \int_0^{B_1} d U_1 \int_{U_1}^{B_2} d U_2 \pr{\forall i=3, \ldots, n: U_i < B_i \text{ and } U_2 < U_3 < \cdots
< U_n | U_2}  \\
    &=
    \cdots
    =
    n! \int_0^{B_1} d U_1 \int_{U_1}^{B_2} d U_2 \int_{U_2}^{B_3} d U_3 \cdots \int_{U_{n-1}}^{B_{n}} d U_n. \label{eq:repeated_integral}
\end{align}
This recursion was first noted by \cite{WaldWolfowitz1939} who demonstrated the symbolic computation of the integral with
$n=6$ samples.
The integral \eqref{eq:repeated_integral} was analyzed by \cite{Durbin1973} for the case where the bounds $B_i$ increase
linearly, leading to a closed-form expression for the distribution of the one-sided Kolmogorov-Smirnov statistics.
More recently, \cite{MoscovichNadlerSpiegelman2016} developed a method for the numerical integration of \eqref{eq:repeated_integral} with computational cost $O(n^2)$.
That method was shown to be stable up to $n \approx 30,000$ using standard double-precision floating-point numbers.

Many other recursive formulas have been proposed for the calculation of one-sided non-crossing probabilities.
Of particular note is the formula in Proposition 3.2 of \cite{DenuitLefevrePicard2003}. 
This $O(n^2)$ recursive formula was first derived by \cite{NoeVandewiele1968} and used to tabulate percentage points of standardized
one-sided Kolmogorov-Smirnov statistics for sample sizes up to $n=100$.
Another $O(n^2)$ recursive procedure was proposed by \cite{KotelnikovaKhmaladze1983}.
A major limitation of these methods is that they contain sums of large binomial coefficients multiplied by very small numbers, leading to numerical instabilities.
Thus, using standard floating-point numbers, the methods become unstable for sample sizes beyond a few hundred (see Section 1 of \cite{KhmaladzeShinjikashvili2001}).
While it is possible to use variable precision floating-point numbers or rational arithmetic to alleviate the loss of numerical
accuracy \citep{BrownHarvey2008a,BrownHarvey2008b,SchroederDickhaus2020}, this approach incurs heavy runtime
penalties compared to the use of numerically stable methods that can use standard floating-point numbers.

\subsection{Two-sided boundaries}
For the computation of two-sided non-crossing probabilities of the form~\eqref{eq:two_sided},
several methods have been proposed \citep{Epanechnikov1968, Steck1971,Durbin1971,Noe1972,FriedrichSchellhaas1998,KhmaladzeShinjikashvili2001,MoscovichNadler2017}.
Unfortunately, all of these methods have a high computational cost of $O(n^3)$ with the exception
of the following:

\begin{itemize}
    \item The FFT-based algorithm of \cite{MoscovichNadler2017}, on which the current paper is based, has a running time of $O(n^2 \log n)$.
    \item The procedure of \citet{Durbin1971} is based on solving a system of linear
equations.
    While standard solutions are $O(n^3)$, using the Coppersmith-Winograd algorithm or related methods,
the theoretical asymptotic runtime is approximately $O(n^{2.373})$.
    However, such methods involve huge runtime constants and are not practical.
 
    \item It was noted by \cite{Miecznikowski2017}  that the determinant-based formula of \cite{Steck1971} can be computed in $O(n^2)$ thanks to the Hessenberg form of the matrix. However, due to a rapid loss of numerical accuracy, this approach is difficult to scale to large values of $n$. In a recent paper by \cite{Wang2022}, the authors compared seven variants of high-precision and rational arithmetic algorithms for computing Steck's determinant.
They demonstrated their approach to the task of computing $p$-values for the one-sided exact Berk-Jones statistic $M_n^+$ (see Eq. \eqref{def:M_n_plus}).
For the largest sample size that  they tested ($n=15,000$)  the running time of computing the probability \eqref{eq:two_sided} was 40 seconds.
In contrast, for the same sample size, the $O(n^2)$ method presented in this paper runs in $0.33$ seconds.
\end{itemize}
% 
% \begin{table}
% \begin{tabular}{cccc}\hline
% Method & Boundary type & Runtime & Scalability \\\hline
% \cite{NoeVandewiele1968}             & one-sided & $O(n^2)$ & $?$ \\\hline
% \cite{KotelnikovaKhmaladze1983}      & one-sided & $O(n^2)$ & $n \approx 120$ \\\hline
% \cite{MoscovichNadlerSpiegelman2016} & one-sided & $O(n^2)$ & $n \approx 30,000$ \\\hline
% This paper & one-sided & $O(n^2)$ & $n > 1,000,000$ \\\hline
% \hline
% \cite{Epanechnikov1968}              & two-sided & $O(n^3)$ & $?$ \\\hline
% \cite{Durbin1971}                    & two-sided & $O(n^{2.373})$ & $?$ \\\hline
% \cite{Steck1971}/\cite{Wang2022}     & two-sided & $O(n^2)$ & $n \approx 100$ \\\hline
% \cite{Noe1972}                       & two-sided & $O(n^3)$ & $n \approx 1800$ \\\hline
% \cite{FriedrichSchellhaas1998}       & two-sided & $O(n^3)$ & $?$ \\\hline
% \cite{KhmaladzeShinjikashvili2001}   & two-sided & $O(n^3)$ & $n > 35,000$ \\\hline
% \cite{MoscovichNadler2017}           & two-sided & $O(n^2 \log n)$ & $n > 1,000,000$ \\\hline
% \end{tabular}
% \label{tab:algorithms_comparison}
% \caption{List of methods for computing one-sided and two-sided boundary crossing probabilities of the form \eqref{eq:order_statistics_upperbound_problem} and \eqref{eq:two_sided}, respectively.
% In the \emph{Scalability} column we list (when it is known) the maximum sample size for which the method is applicable using standard IEEE double-precision floating point numbers, using data from \cite{KhmaladzeShinjikashvili2001,Wang2022} and our own experiments.}
% \end{table}

%%%%%%%%%%%%%%%%%%%%%%%%%%%%%%%%%%
\section{Technical background} \label{sec:background}
%%%%%%%%%%%%%%%%%%%%%%%%%%%%%%%%%%

We now describe the methods of \cite{FriedrichSchellhaas1998, KhmaladzeShinjikashvili2001,
MoscovichNadler2017} that form the basis of our algorithm.
These methods compute the two-sided non-crossing probability \eqref{eq:two_sided} given a set of lower and upper boundaries.
However, since the focus of this paper is on the one-sided case, our exposition describes these methods  in the simpler case of a one-sided boundary, where $b_i = 0$ for all $i$.

\subsection{Stepwise Binomial propagation} \label{sec:stepwise_binomial}
In this subsection, we describe a minor variant of ``scheme 1'' of \cite{FriedrichSchellhaas1998}, specialized  to the one-sided
boundary case.
Let $S(i,j)$ be the following probability,
\begin{align} 
    S(i,j)
    &:=
    \pr{n F_n(B_i) = j \text{ and } U_{(1)} \le B_1, U_{(2)} \le B_2, \ \cdots,\  U_{(i-1)} \le B_{i-1}}  \\
    &=
    \pr{n F_n(B_i) = j \text{ and } \forall \ell \in \{1,2,\ldots,i-1\}:  n F_n(B_\ell) \ge \ell}, \label{def:Sij}
\end{align}
where $U_{(1)} \le \cdots \le U_{(n)}$ are the order statistics of a sample $U_1, \ldots, U_n \simiid U[0,1]$.
Define
\begin{align}
    B_0 := 0 \ \text{ and }\  B_{n+1} := 1.
\end{align}
With this notation,
\begin{align}
    S(0, j) = \pr{n F_n(0) = j} = \delta_{0,j}.
\end{align}
Our quantity of interest is
\begin{align}
    S(n+1, n)
    &=
    \pr{ n F_n(1) = n \text{ and } \forall i \in \{1,\ldots,n\}:\ U_{(i)} \le B_i} \\
    &=
    \pr{\forall i \in \{1,\ldots,n\}:\ U_{(i)} \le B_i} \\
    &=
    \text{NCPROB}(B_1, \ldots, B_n).
\end{align}
We now explain how this quantity is computed using recursion relations.
The initial conditions are $S(0,j) = \delta_{0,j}$.
The transition probabilities are given by the following Chapman-Kolmogorov equations,
\begin{align} \label{eq:Sij}
    S(i+1,j)
    =
    \sum_{k = i}^n S(i, k) \cdot \pr{n F_n(B_{i+1}) = j \big| n F_n(B_i) = k}.
\end{align}
Note that the summation is done over $k \ge i$ to guarantee that we only sum over  non-crossing paths for which $n F_n(B_i) \ge i$, or equivalently that $U_{(i)} \le B_{i}$ (see  Figure~\ref{fig:one_sided_crossing_illustration} and \ref{appendix:reduction}).
The transition probability $k \to j$ is the probability that exactly $j-k$ of the points $U_1, \ldots, U_n$ fall in the interval $(B_i, B_{i+1}]$, conditioned on the fact that $k$ of them fell in the interval $[0, B_i]$.
This is given by the following  Binomial probability mass function,
\begin{align} \label{eq:prob_k_ECDF_jumps}
    \pr{n F_n(B_{i+1}) = j \big| n F_n(B_i) = k}
    &=
    \pr{\text{Binomial}(n-k, p_i) = j-k} \\
    &=
    {n-k \choose j-k} p_i^{j-k} (1-p_i)^{n-j}.
\end{align}
where $p_i :=  \pr{B_i < U \le B_{i+1}| U > B_i}$ with $U \sim U[0,1]$. Hence, $p_i = (B_{i+1}-B_i)/(1-B_i)$.
To compute the non-crossing probability $S(n+1, n)$, one can start by setting $S(1,j)$ for all $j$ via Eq. \eqref{eq:Sij},
then proceed
to compute $S(2,j)$ for all $j$, etc. at a total runtime cost of $O(n^3)$.
This procedure, which we dub \emph{Stepwise Binomial propagation}, is illustrated in Figure \ref{fig:one_sided_crossing_illustration}.
The filled circles represent elements of $S(i,j)$ with $j \ge i$ whereas the hollow circles $S(i,i-1)$ correspond to paths for which $n F_n(t)$ crosses the lower boundary at $B_i$.

\subsection{Stepwise Poisson propagation} \label{sec:stepwise_poisson}

There is a simple connection between the empirical CDF of an i.i.d. sample and a conditioned Poisson process:
\begin{lemma} \label{lemma:equiv}
    Let $U_1, \ldots, U_n \simiid U[0,1]$ be a sample and let $F_n(t) = \tfrac1n \sum_i \mathbf{1}(U_i
\le t)$ be its empirical CDF. The distribution of the process $n F_n(t)$ is identical to that of a Poisson process
$\xi_n(t)$ with intensity
$n$, 
    conditioned on $\xi_n(1) = n$.
\end{lemma}
For the proof, see \citet[Prop 2.2, Ch. 8]{ShorackWellner2009}.
The calculation of the non-crossing probability in Eq. \eqref{eq:pr_binomial_no_cross} may thus be reduced to the calculation
of the non-crossing probability of a Poisson process $\xi_n$ with intensity $n$.
Let $Q(i,j)$ be the non-crossing-up-to-$B_i$ probabilities of $\xi_n$, defined in analogy to $S(i,j)$ in
Eq. \eqref{def:Sij},
\begin{align} \label{def:Q}
    Q(i,j) := \pr{\xi_n(B_i) = j  \text{ and } \forall \ell \in \{1,2,\ldots,i-1\}:\ \xi_n(B_\ell) \ge \ell}.
\end{align}
The recursion relations for all $i=0, \ldots, n+1$ and $j=0,\ldots,n$ mimic  those of $S(i,j)$,
\begin{align} 
    Q(0,j) &= \delta_{0,j}, \label{eq:Qij_forward_equations} \\
    Q(i+1,j) &= \sum_{k=i}^n Q(i, k) \cdot \pr{\xi_n(B_{i+1})=j \big| \xi_n(B_i) = k}, \label{eq:Qij_transition}
\end{align}
where the transition probabilities are now given by Poisson counts,
\begin{align} \label{eq:poisson_transition_probabilities}
    \pr{\xi_n(B_{i+1})=j \big| \xi_n(B_i) = k}
    &=
    \pr{\text{Pois}(n(B_{i+1}-B_{i})) = j-k} \\
    &=
    \frac{(n(B_{i+1}-B_{i}))^{j-k} e^{-n(B_{i+1}-B_{i})}}{(j-k)!}.
\end{align}
As before, in Equation \eqref{eq:Qij_transition} we sum  over $k \ge i$ to guarantee that we only consider the non-crossing paths for which $\xi_n(B_i) \ge i$. 
The algorithm based on this recursion, which we dub \emph{stepwise Poisson propagation} proceeds by computing $Q(1,j)$ for
all $j$, then $Q(2,j)$ for all $j$, etc.
Finally, by Lemma \ref{lemma:equiv},

\begin{align} \label{eq:binomial_noncrossing_problem}
    S(n+1, n)
    &=
    \pr{n F_n(B_1) \ge 1, \ \ldots,\  n F_n(B_{n}) \ge n} \\
    &=
    \pr{\xi_n(B_1) \ge 1, \ \ldots,\  \xi_n(B_{n}) \ge n | \xi_n(1) = n} \\
    &=
    \frac{\pr{\xi_n(B_1) \ge 1, \ \ldots,\  \xi_n(B_{n}) \ge n \text{ and } \xi_n(1) = n}}{\pr{\xi_n(1) = n}} \\
    &=
    \frac{\pr{\xi_n(B_1) \ge 1, \ \ldots,\  \xi_n(B_{n}) \ge n \text{ and } \xi_n(B_{n+1}) = n}}{\pr{\text{Pois}(n)=n}} \\
    &=
    \frac{Q(n+1,n)}{n^n e^{-n}/n!}.
\end{align}
This method was proposed by \citet{KhmaladzeShinjikashvili2001} for  two-sided boundary crossing probabilities.
It has the same $O(n^3)$ asymptotic running time as the stepwise Binomial propagation of \cite{FriedrichSchellhaas1998} which
we described in Section \ref{sec:stepwise_binomial}.

\subsection{Fourier-based stepwise Poisson propagation} \label{sec:mn2017}
In contrast to the Binomial propagation described in Section \ref{sec:stepwise_binomial}, in the stepwise Poisson propagation,
the transition probabilities $\pr{\xi_n(B_{i+1})=j \big| \xi_n(B_{i}) = k} $ in Eq. \eqref{eq:poisson_transition_probabilities} do not depend on $j$ or $k$ but only on their
difference.
This is due to the memorylessness property of the Poisson process.
As a result, the recurrence  \eqref{eq:Qij_transition} has the form of a linear convolution,
\begin{align} \label{eq:Q_iplus1_j}
        Q(i+1,j) &= \sum_{k=i}^n Q(i, k) \cdot \pr{\text{Pois} ( \lambda^{(i+1)} ) = j-k},
\end{align}
where $\lambda^{(i+1)} := n(B_{i+1}-B_{i})$ is the expected number of jumps of the Poisson process in the interval $(B_{i+1},B_i]$.
Let zero$({\bf v}, i)$ denote a copy of the vector $\bf v$ with the first $i$ elements set to zero,
and let $Q^{(i)} \in \mathbb{R}^{n+1}$ denote the vector
\begin{align} \label{eq:Q_vector}
    Q^{(i)} :=\left( Q(i,0),  Q(i,1),\ldots, Q(i,n) \right).
\end{align}
With this notation,
the vector $Q^{(i+1)}$ is given by a truncated linear convolution,
\begin{align} \label{eq:Q_next}
    Q^{(i+1)} &= \text{zero}(Q^{(i)}, i) \star \pi^{(i+1)},
\end{align}
where \( \pi^{(i+1)} := \left( \pr{\text{Pois}(\lambda^{(i+1)}) = 0}, \ldots, \pr{\text{Pois}(\lambda^{(i+1)}) = n} \right) \) is the Poisson PMF vector.
The zeroing operation is done to account for the fact that the summation in Eq. \eqref{eq:Q_iplus1_j} is performed only for $k \ge i$. 

Each of these linear convolutions can be computed efficiently in $O(n \log n)$ time steps using the fast Fourier transform
(FFT) and the circular convolution theorem for discrete signals \citep[Ch. 12, 13]{NumericalRecipes1992}.
The resulting procedure has a total running time of $O(n^2 \log n)$ and is numerically stable for large sample sizes using
standard double-precision (64-bit) floating-point numbers \citep{MoscovichNadler2017}.

This stepwise FFT-based procedure can also be used to compute the non-crossing probabilities for non-homogeneous
Poisson processes, negative binomial processes, and other types of stochastic jump processes, as well as non-crossing probabilities for discontinuous distributions~\citep{Dimitrova2020,Dimitrova2020a}.

% 
% \begin{algorithm}
%     \caption{Stepwise Poisson propagation \citep{KhmaladzeShinjikashvili2001}}\label{euclid}
%     \begin{algorithmic}[1]
%         \Function{Noncrossing-Poisson}{$t_{-1}, \ldots, t_N$}
%             \State $Q[0] \gets 1$
%             \State $Q[1], \ldots, Q[n] \gets 0$
%             \For{$i \gets 0 \ldots, N$}
%                 \State $\lambda \gets n(t_{i}-t_{i-1})$
%                 \For{$j \gets i+1, \ldots, n$}
%                     \State $Q'[j] \gets \sum_{k=0}^j  \frac{\lambda^{j-k} e^{-\lambda}}{(j-k)!} \cdot Q[k]   $
%                 \EndFor
%                 \State $Q \gets Q'$
%             \EndFor
%             \State \textbf{return} $Q[n] / \frac{n^n e^{-n}}{n!}$
%         \EndFunction
%     \end{algorithmic}
% \end{algorithm}
% 
%%%%%%%%%%%%%%%%%%%%%%%%%%%%%%%%%%
\section{Proposed algorithm} \label{sec:method}
%%%%%%%%%%%%%%%%%%%%%%%%%%%%%%%%%%

\label{sec:n2_one_sided}
In the previous section, we described how, for any $i$, one can obtain the non-crossing probabilities vector $Q^{(i+1)}$, defined
in Eq. \eqref{eq:Q_vector}, by computing a truncated linear convolution of $Q^{(i)}$ and the PMF of a Poisson random variable.
Starting from $Q^{(i)}$ for some $i$ and repeating this process $k$ times we obtain $Q^{(i)} \to Q^{(i+1)} \to \ldots \to Q^{(i+k)}$ in
$O(k n \log n)$ time.
In this section, we show how for any $k \in (\log n, n/\log n)$, it is possible to go directly from $Q^{(i)}$ to $Q^{(i+k)}$
using just $O(kn)$ steps. 
This makes the total runtime for computing  $Q^{(n+1)}$ be 
\(
    \left \lceil \tfrac{n+1}{k} \right \rceil O(kn) = O(n^2).
\)
The main idea behind our method is simple.
We first compute all the transition probabilities of the Poisson process $\xi_n$,
\begin{align}
    Q(i,j) \cdot \pr{\xi_n(B_{i+k}) = \ell | \xi_n(B_i)=j}, \label{eq:poisson_j_to_l}
\end{align}
from all non-zero elements $Q(i,j)$ at a cost of $O(n \log n)$ using a single convolution.
These transition probabilities  include the contributions of  non-crossing paths and also the contributions of crossing paths that
intersect the lower boundary in the interval $[B_i, B_{i+k})$.
All that remains is to subtract the contributions of the crossing paths.
Non-crossing paths satisfy $\xi_n(B_j) \ge j$
for all $j$. In contrast, a path that crosses the lower boundary inside the interval $[B_i, B_{i+k})$ must satisfy $\xi_n(B_j) < j$ for at least one index $j \in \{i, \ldots, i+k-1\}$.
With some careful accounting that we describe in the next section, we can efficiently compute the probability of having a
first crossing at each of these points and then subtract their individual contributions from the arrival probabilities in Eq.~\eqref{eq:poisson_j_to_l}.

\begin{definition}
    Let $f:[0,1] \to \{0,1,2, \ldots\}$ be a function. For every $i \in \{0, \ldots, n+1\}$, we define two logical predicates,    
    \begin{align}
        \text{NC}(f,i)
        &:=
        \forall \ell < i: f(B_{\ell}) \ge \ell, && \text{(no crossing before $B_i$)} \\
        \text{FC}(f,i)
        &:=
        \text{NC}(f,i) \text{ and } f(B_i) < i && \text{(first crossing at $B_i$)} \\
        &=
        \forall \ell < i: f(B_{\ell}) \ge \ell \text{ and } f(B_i) = i-1. \label{eq:FC}
    \end{align}
\end{definition}
\noindent Q satisfies by its definition \eqref{def:Q},
\begin{align}
    Q(i+k,j)
    &=
    \pr{\xi_n(B_{i+k}) = j  \text{ and } \text{NC}(\xi_n,i+k)}.
\end{align}

\begin{proposition} \label{prop:Q_star}
    Let $k > 0$ be some integer. Given the vector $Q^{(i)}$ as defined in Eq. \eqref{eq:Q_vector}, we can compute the probabilities
    \(
        \pr{\xi_n(B_{i+k}) = j \text{ and } \emph{NC}(\xi_n, i)}
    \)
    for all $j$ in $O(n \log n)$ time.
\end{proposition}
\begin{proof}

\begin{align}
   \pr{\xi_n(B_{i+k}) = j \text{ and } \emph{NC}(\xi_n, i)}
    =
    \sum_{i \le k \le j} Q(B_i, k) \cdot\pr{\text{Pois}(n(B_{i+k}-B_i)) = j-k}.
\end{align}
These values, as a function of $j$, are a truncated linear convolution of $Q^{(i)}$ and the PMF of
a Poisson random variable with intensity $n(B_{i+k}-B_i)$.
As explained in Section \ref{sec:mn2017}, this convolution can be computed in $O(n \log n)$ steps.
\end{proof}
\begin{proposition} \label{prop:firstcross_and_m}
    Given $Q^{(i)}$ the probabilities  $\pr{\text{FC}(\xi_n, j) \text{ and } \xi_n(B_{i+k}) = \ell}$ for all values of $\ell \in \{i+k, \ldots, n\}$ and $j \in \{i, \ldots,i+k-1\}$ can be computed in $O(k^2 \log k + nk)$ time.
\end{proposition}
\begin{proof}
We first note that for every $j \ge 1$, by Eq. \eqref{eq:FC},
\begin{align} \label{eq:firstcross_Q}
    \text{FC}(\xi_n,j)
    =
    \forall \ell < j-1: f(B_{\ell}) \ge \ell \text{ and } f(B_{j-1}) = f(B_j) = j-1.
\end{align}
By the chain rule, we have
\begin{align} 
    &\pr{\text{FC}(\xi_n,j)} \\
    &=
    \pr{\forall \ell < j-1: f(B_{\ell}) \ge \ell \text{ and } f(B_{j-1}) = j-1}
    \cdot \pr{f(B_j) = j-1 | f(B_{j-1}) = j-1} \nonumber \\
    &=
    Q(j-1,j-1) \cdot \pr{\text{Pois}\left(n(B_j - B_{j-1})\right) = 0}.
\end{align}
From the definition of FC, if FC$(\xi_n,j)$ then $\xi_n(B_j) = j-1$, hence by the memorylessness of the Poisson process,
\begin{align}
    \pr{\xi_n(B_{i+k}) = \ell \vert \text{FC}(\xi_n, j)} 
    &=
    \pr{\xi_n(B_{i+k})=\ell | \xi_n(B_j) = j-1}
    \\
    &=
    \pr{\text{Pois}(n(B_{i+k} - B_j)) = \ell-(j-1)}.
\end{align}
Putting it all together, we have
\begin{align}
    &\pr{\text{FC}(\xi_n, j) \text{ and } \xi_n(B_{i+k}) = \ell}
    =
    \pr{\text{FC}(\xi_n, j)} \cdot \pr{\xi_n(B_{i+k}) = \ell | \text{FC}(\xi_n, j)} \\
    &=
    Q(j-1,j-1) \cdot \pr{\text{Pois}\left(n(B_j - B_{j-1})\right) = 0} \cdot \pr{\text{Pois}(n(B_{i+k} - B_j)) = \ell-j+1}.
\nonumber
\end{align}
Evaluating this probability for all $j \in \{i, \ldots, i+k-1\}$
and
 $\ell \in \{i+k, \ldots, n\}$
takes a total of $O(nk)$ time.
As for the computation of $Q(j-1, j-1)$ for all $j \in \{i, \ldots, i+k-1\}$, note that  $Q(j,\ell)$ for all $j,\ell \in \{i-1, \ldots, i+k-1\}$ is a $k \times k$ sub-array that can be computed in time $O(k^2 \log k)$ using the FFT-based algorithm described in Section \ref{sec:stepwise_poisson}.
\end{proof}
\begin{proposition} \label{prop:Q_t_iplusk_from_Q_t_i}
    Given $Q^{(i)}$, one can compute $Q^{(i+k)}$ in $O(n \log n + nk +\ k^2 \log k)$ time.
\end{proposition}
\begin{proof}
If the predicate NC$(f,i)$ holds then either NC$(f,i+k)$ or FC$(f,j)$ for exactly one of $j \in \{i,\ldots,i+k-1\}$.
Hence for a Poisson process $\xi_n(t)$
\begin{align} \label{eq:just_some_eq}
    \pr{\text{NC}(\xi_n, i)} = \pr{\text{NC}(\xi_n, i+k)} + \sum_{j=i}^{i+k-1} \pr{\text{FC}(\xi_n, j)}.
\end{align}
This equality holds  even when we add the constraint that $\xi_n(B_{i+k}) = \ell$.
Adding this constraint and subtracting the sum on the RHS of Eq. \eqref{eq:just_some_eq} from both sides, we get,
\begin{align}
    \label{eq:conditioned_nocross}
    &\pr{\text{NC}(\xi_n, i+k) \text{ and } \xi_n(B_{i+k}) = \ell} \nonumber \\
    &= \underbrace{\pr{\text{NC}(\xi_n, i) \text{ and } \xi_n(B_{i+k}) = \ell}}_{(*)} - \sum_{j=i}^{i+k-1} \underbrace{\pr{\text{FC}(\xi_n,
j) \text{ and } \xi_n(B_{i+k}) = \ell}}_{(**)} .
\end{align}
By Proposition \ref{prop:Q_star} the probabilities (*) can be computed in time $O(n \log n)$
and by Proposition \ref{prop:firstcross_and_m}  the probabilities (**) are computable in time $O(nk + k^2 \log k)$.
Evaluating \eqref{eq:conditioned_nocross} costs $O(nk)$.
The total running time of computing $Q^{(i+k)}$ given $Q^{(i)}$ is therefore $O(n \log n + nk + k^2 \log k)$.
\end{proof}

We can now put it all together.
    Starting from $Q^{(0)} = (1,0,0,\ldots,0)$, we compute $Q^{(k)}$ 
    and then compute $Q^{(2k)}$, $Q^{(3k)}$, etc., until we reach $Q^{(n+1)}$. By Proposition \ref{prop:Q_t_iplusk_from_Q_t_i},
each of
these steps takes $O(n \log n +\ nk + k^2 \log k)$ time.
    Thus the total running time is
    \begin{align} \label{eq:runtime_given_k}
        \left \lceil \tfrac{n+1}{k} \right \rceil  O \left( n \log n +\ nk + k^2 \log k \right)
        =
        O \left( \tfrac{n^2 \log n}{k} + n^2 + nk \log k \right).
    \end{align}
For any choice of $k \in (\log n, n / \log n)$, the running time is $O(n^2)$.

\section{Benchmarks} \label{sec:benchmarks}
In this section, we test the running time and accuracy of our method.
The application chosen here is the computation of $p$-values for the $M_n^+$ one-sided statistic of \cite{BerkJones1979} as defined in Eq. \eqref{def:M_n_plus}.
Following the work of \cite{DonohoJin2004}, the Higher Criticism and Berk Jones statistics have attracted renewed
interest due to their optimality with respect to various sparse signal detection problems~\citep{HallJin2010,AriasCastroCandesPlan2011,LiSiegmund2015,AriasCastroHuangVerzelen2020,PorterStewart2020,ZhangJinWu2020,Kipnis2022}.
In particular, the $M_n^+$ and closely related $R_n^+$ statistics of \cite{BerkJones1979} have been applied to inference tasks in various domains, including survival analysis, astrophysics, genetics, and social network anomaly detection \citep{Owen1995,SulisMaryBigot2017,SunLin2019,Zhang2022,Matthews2013,Cadena2019}.

For each sample size $n$, we first computed an $\alpha$-level threshold $s_{n,\alpha}$ with the bisection method described
in Section~\ref{sec:threshold} for $\alpha=5\%$. The bounds $B_i=s_i^{-1}(s)$ were computed using the \texttt{betaincinv} function, which computes the inverse of the CDF of a Beta random variable. The probability  $\pr{M_n^+ < s_{n,\alpha}}$ was calculated using a single  NCPROB evaluation as described in Section~\ref{sec:pvalue}. The following methods for computing NCPROB were tested:
\begin{itemize}
    \item {\bf KS (2001):} the $O(n^3)$ two-sided algorithm of \cite{KhmaladzeShinjikashvili2001}, described in Section \ref{sec:stepwise_poisson}.
    \item {\bf MNS (2016):} the $O(n^2)$ one-sided algorithm of \cite{MoscovichNadlerSpiegelman2016} mentioned in Section
\ref{sec:existing}.
    \item {\bf MN (2017):} the $O(n^2 \log n)$ two-sided algorithm of \cite{MoscovichNadler2017}, described in Section \ref{sec:mn2017}.
    \item {\bf New:} the $O(n^2)$ one-sided algorithm described in this paper.
\end{itemize}
Figure \ref{fig:benchmarks} shows the running times for the sample sizes $n=5000,\ 10000, \ldots, 100,000$ (best out of 3 runs).
The Poisson-propagation-based methods {\bf KS\,(2001)/MN\,(2017)/New} all produce the same results in the tested range, with
relative errors less than $10^{-10}$ using standard double-precision (64-bit) floating-point numbers.
In contrast, {\bf MNS (2016)} is only accurate up to about $n=30,000$. For the sample size $n=35,000$ it produces a relative error of $7\%$
and for $n>50,000$ it breaks down completely. Therefore, we did not test the running time of {\bf MNS (2016)} for sample
sizes larger than $30,000$.

\begin{figure}
    \vspace*{-12mm}
        \hspace*{-2.9mm}
    \includegraphics[width=15.03cm]{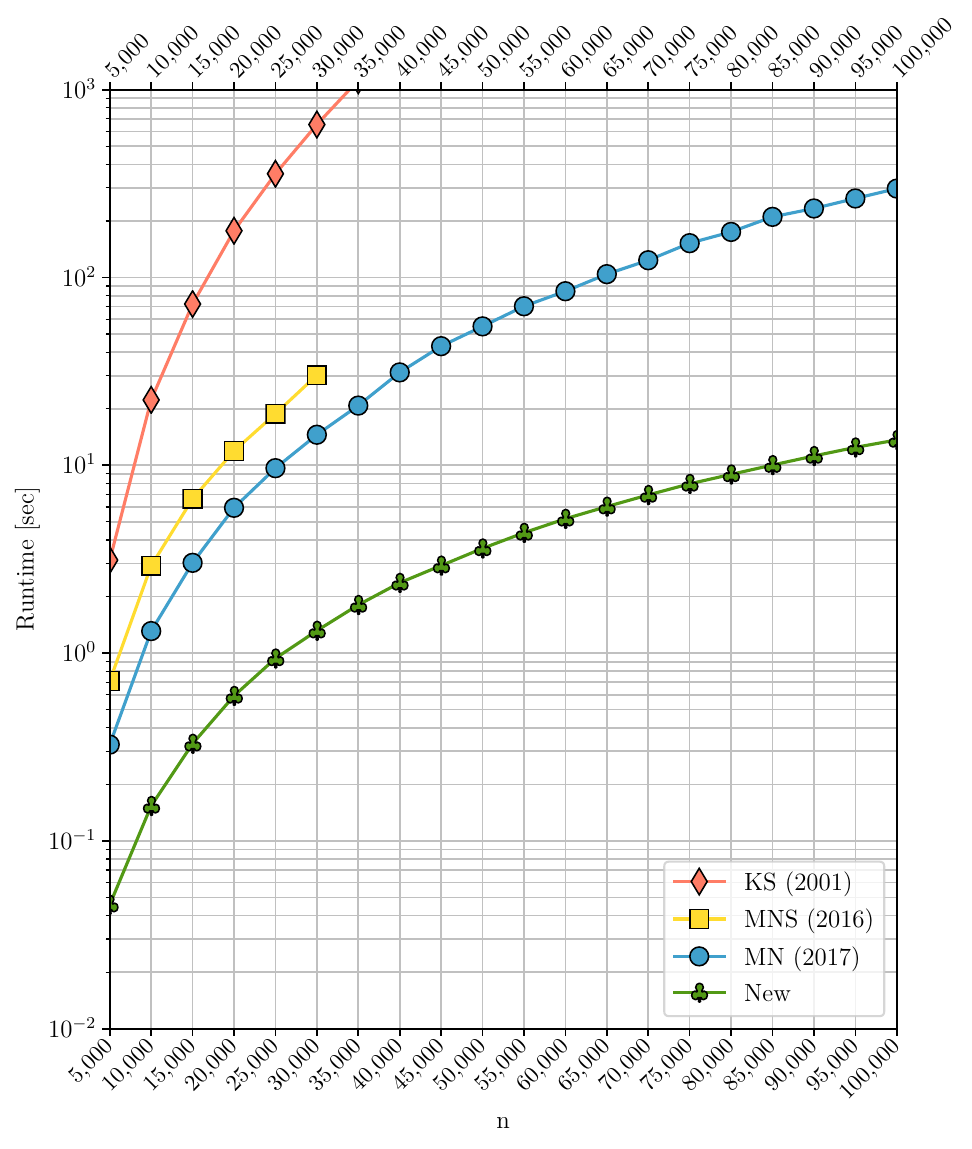}
    \vspace*{-12mm}
    \caption{\label{fig:benchmarks}
    Running times for computing the $p$-value of a one-sided goodness-of-fit statistic. The times shown are the best out of three runs.
Note the logarithmic y-axis.
    %MNS (2016) exhibits rapid loss of accuracy for $n > 30,000$ so we did not run it on larger sample sizes.
} 
\end{figure}
\begin{figure}
    %\centering
    \vspace*{-12mm}    
    \includegraphics[width=15cm]{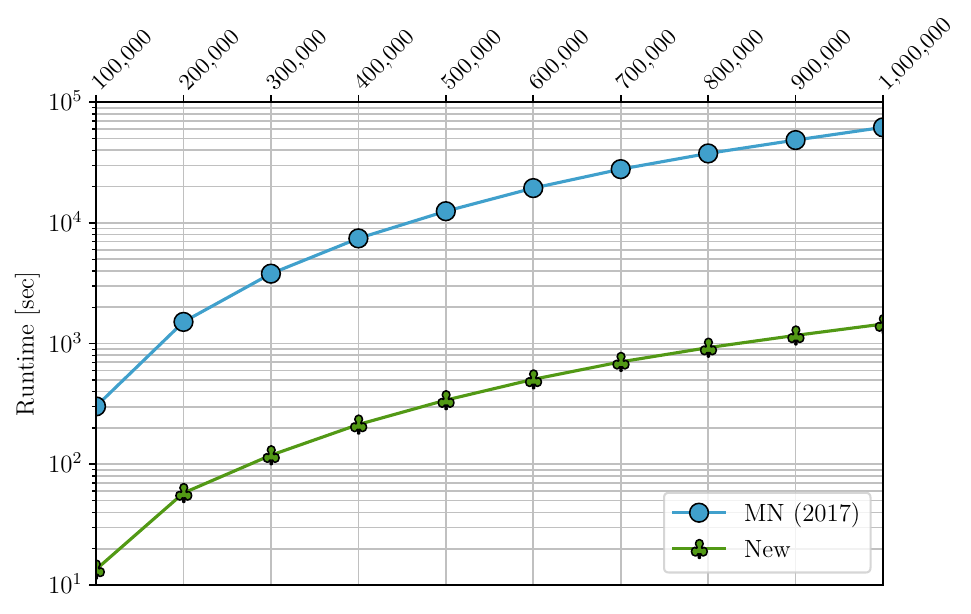}
    \hspace*{-0.1mm}
    \includegraphics[width=13.88cm]{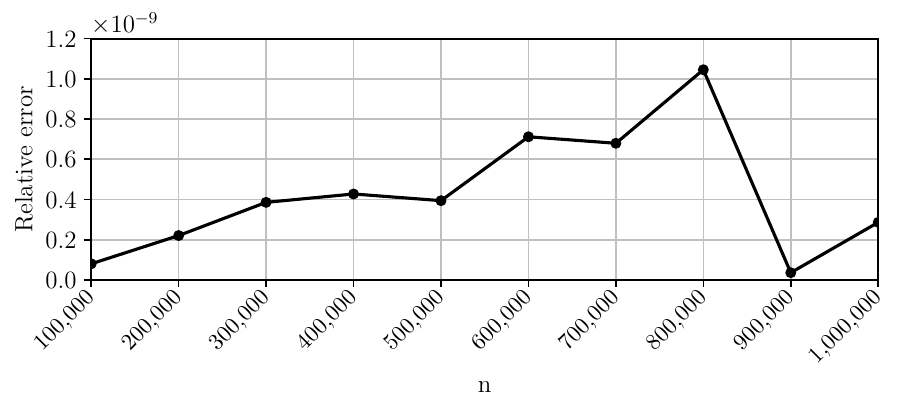}
    %\vspace*{-12mm}
    \caption{Large scale benchmarks for computing the $p$-value of a one-sided goodness-of-fit statistic. \qquad {\bf (top)}~Running times of our method vs. \cite{MoscovichNadler2017}. Note the logarithmic
y-axis. {\bf (bottom)} Relative numerical errors of the $p$-values computed by the two methods. The maximum relative numerical error is about $10^{-9}$.   \label{fig:benchmarks_large}}
\end{figure}

An additional set of  large-scale benchmarks is shown in Figure \ref{fig:benchmarks_large}.
This time, due to the long running times involved, we only performed a single measurement for every data point (rather than
taking the best out of 3 runs) and used a fixed threshold for all sample sizes, chosen to be the $\alpha$-level threshold
for $n=100,000$ with $\alpha=5\%$.
This figure does not show benchmarks for {\bf KS (2001)} due to its excessive running time for  large sample sizes.
In the bottom panel, we show the relative difference between the boundary-crossing probabilities computed using {\bf MN (2017)} and {\bf New}.
This relative error is small throughout the tested range.
See \ref{sec:implementation} for additional details on our implementation and benchmarks. 

\begin{remark}
    A different choice of test statistic should yield very similar running times.
    The chosen test statistic and threshold merely determine the bounds $B_1, \ldots, B_n$, but the running time does not typically depend on their particular values.
One exception is the case where there are multiple repeating bounds (e.g. $B_1=B_2$), which we optimized for.
\end{remark}

\subsection{Code availability} \label{sec:code}
A  C++ implementation of the tested methods  for computing one-sided and two-sided
boundary crossing probabilities of the form
\eqref{eq:order_statistics_upperbound_problem} and \eqref{eq:two_sided} is provided at the following link:

\url{https://github.com/mosco/crossing-probability}

\noindent This repository also includes a Python language wrapper, as well as code for running the benchmarks and creating the figures
in Section~\ref{sec:benchmarks}.

\section{Conclusion}

Given a set of bounds $B_1, \ldots, B_n \in [0,1]$, this paper presents a new $O(n^2)$ method for the calculation of the  non-crossing probability
\begin{align}
    \text{NCPROB}(B_1, \ldots, B_n) := \pr{\forall i: U_{(i)} \le B_i \big},
\end{align}
where $U_{(1)} \le \cdots \le U_{(n)}$ are the order statistics of a uniform draw in the unit interval.
The fast calculation of these probabilities has many applications, in particular for sparse signal detection, goodness-of-fit testing, financial risk modeling, and the construction of one-sided confidence bands for the empirical distribution function.

We have applied our method to the computation of $p$-values for a one-sided goodness-of-fit statistic of \citet{BerkJones1979} and compared its running time to other leading methods, with sample sizes as large as one million. For all sample sizes, our method is shown to be the fastest one available by a wide margin.

\subsection*{Acknowledgments}
Some of this research was done while the author was a postdoctoral research associate at the Program in Applied and Computational Mathematics (PACM), Princeton University.
The author is supported by an Israel Science Foundation grant (1662/22). 
%\FloatBarrier
%\bigskip

\appendix

\section[Appendix A. Reduction of the continuous boundary crossing problem to a discrete set of inequalities]{Reduction of the continuous boundary crossing problem to a discrete set of inequalities} \label{appendix:reduction}
% 
% We first note that in the continuous boundary-crossing problem \eqref{eq:pr_binomial_no_cross}, the boundary function $b(t)$
% may be
% assumed to be monotone non-decreasing without loss of generality.
% This is a direct consequence of the following result,
% \begin{proposition} \label{prop:monotone_wlog}
%     Let $b^*(t) := \sup_{u \in [0,t]} b(u)$.
%     Then:
%     \begin{enumerate}[(i)]
%         \item $b^*$ is a non-decreasing function.
%         \item $\forall t \in [0,1]:\ b(t) \le n F_n(t) \quad \Longleftrightarrow \quad \forall t \in [0,1]:\ b^*(t)\le
% n F_n(t).$
%     \end{enumerate}
% \end{proposition}
% \begin{proof}
% (i) Trivial. (ii) For the $\Longleftarrow$ direction, we simply note that   
% $b \le b^*$ and therefore $b \le b^* \le  n F_n $.
% The $\Longrightarrow$ direction follows from the monotonicity of $F_n$,
% \begin{align}
%     b^*(t) = \sup_{u \in [0,t]} b(u) \le \sup_{u \in [0,t]} n F_n(u) = n F_n(t).
% \end{align}
% \end{proof}
Let $U_1, \ldots, U_n \sim U[0,1]$ be a sample of independent uniform variables
with order statistics $U_{(1)} \le \cdots \le U_{(n)}$ and empirical cumulative distribution function $F_n(t) = \tfrac{1}{n}
\sum_{i=1}^n 1(U_i \le t)$.
In this appendix, we present the reduction between the non-crossing probability of the empirical cumulative distribution,
\begin{align} \label{eq:continuous_noncrossing_prob}
    \pr{\forall t \in [0,1]: b(t) \le n F_n(t)},
\end{align}
and the simultaneous non-crossing probability of the order statistics,
\begin{align} \label{eq:order_statistics_upperbound_problem_2}
    \text{NCPROB}(B_1, \ldots, B_n)
        &= \pr{\forall i: U_{(i)} \le B_i }.
\end{align}
This reduction is well-known and has also been extended to discontinuous distributions \citep{Steck1971,Gleser1985,Dimitrova2020a}.
Nonetheless, we thought it would benefit the reader to include a concise proof of this basic result, which is at the foundation
of the methods described in this paper.
First, we show that, rather than considering the entire boundary function $b(t)$, it suffices to consider its first integer
passage times,
\begin{align} \label{eq:t_i_def}
    B_{i} &:= \inf\{t \in [0,1] : b(t) > i-1 \}, \qquad i=1,\ldots,k
\end{align}
where $k$ is the largest integer for which the set $\{t : b(t) > k-1 \}$ is non-empty.
The following lemma holds the key observation that allows one to replace the infinite set of inequality constraints $\forall
t \in [0,1]:\ b(t) \le nF_n(t)$ with a finite set of inequalities.
\begin{lemma} \label{lemma:reduction_to_finite}
    Let $f:[0,1] \to \{0,1,2, \ldots\}$ be a non-decreasing right-continuous  function and let $b:[0,1] \to \mathbb{R}$
be a function  with  first integer crossings $B_1, \ldots, B_k$, then
\[
    \forall i: f(B_i) \ge i
    \quad
    \Longleftrightarrow
    \quad
    \forall t : b(t) \le f(t).
\]
\end{lemma}
\begin{proof}
\noindent ($\Longrightarrow$)
Divide the interval $[0,1]$ into a disjoint union,
\begin{align} \label{eq:disjoint_union_t_i_b}
    [0,B_1) \cup [B_1, B_2) \cup \ldots \cup [B_{k-1}, B_k) \cup [B_k, 1].
\end{align}
We now prove that $b \le f$ in each of these intervals: 
\begin{enumerate}
    \item $[0,B_1)$: By the definition of $B_i$, for all $t < B_1$ we have $b(t) \le 0,$ and since $f$ is non-negative it
follows
that $b(t) \le 0 \le f(t)$.

    \item $[B_i,B_{i+1})$:
    If $t < B_{i+1}$ then $b(t) \le i$.
    Since we assumed $f(B_i) \ge i$ it follows that for all $t \in [B_i, B_{i+1})$ we have \( b(t) \le i \le f(B_i) \le f(t)
\), where the last inequality is due to the monotonicity of $f$.
    \item $[B_k, 1]$: $f(t) \ge f(B_k) \ge k \ge b(t)$. The first inequality follows from the monotonicity of $f$, the second
is an assumption of the lemma, and the last inequality follows from the definition of $k$.
\end{enumerate}

\noindent ($\Longleftarrow$)
By the definition of $B_i$ there is a series of real numbers $t_j \in [B_i, 1]$ such that $t_j \to B_i$
and $b(t_j) > i-1$. By the assumption  $f(t_j) \ge b(t_j) > i-1$.
Since $i$ and $f(t_j)$ are both integers, this means that $f(t_j) \ge i$.
From the right-continuity of $f$ we conclude that
\begin{align}
    f(B_i) = \lim_{j \to \infty} f(t_j) \ge i.
\end{align} 
\end{proof}
A direct consequence of this lemma is that the probability that the empirical CDF of a uniform sample does not cross a lower
boundary is equal to the probability that the order statistics satisfy a set of simultaneous upper bounds.
\begin{corollary} \label{cor:reduction_to_finite}
    Let $U_1, \ldots, U_n \simiid U[0,1]$ be a sample with empirical cumulative distribution $F_n(t) = \tfrac{1}{n}
\sum_i 1(U_i \le t)$ and let  $b:[0,1] \to \mathbb{R}$ be a function, then
    \begin{align} \label{eq:cor:reduction_to_finite}
        \pr{\forall t \in [0,1]: b(t) \le n F_n(t)}
        =
        \pr{\forall i \in \{1, \ldots, k\}: U_{(i)} \le B_{i}}, \nonumber
    \end{align}
    where $U_{(1)} \le \cdots \le U_{(n)}$ are the order statistics of the sample and $B_1, \ldots, B_k$ are the first
integer crossings of $b(t)$ as defined in Eq. \eqref{eq:t_i_def}.
\end{corollary}
\begin{proof}
    By Lemma \ref{lemma:reduction_to_finite}, $b(t) \le n F_n(t)$ for
all $t$ if and only if $n F_n(B_i) \ge i$ for all i.
    By the definition of the empirical CDF, $n F_n(B_i) \ge i$ if and only if at least $i$ elements of the sample are
 at most $B_i$, in other words, that $U_{(i)} \le B_i$.
 The result follows.
\end{proof}
Hence the problem of computing the non-crossing probability~\eqref{eq:pr_binomial_no_cross}, is reduced to
 the probability that the inequalities $n F_n(B_i) \ge i$ hold at a finite set of times.
The reduction can also be made in the other direction, from the calculation of the discrete
boundary crossing probability \eqref{eq:order_statistics_upperbound_problem_2} to the continuous boundary
crossing \eqref{eq:continuous_noncrossing_prob}.
\begin{corollary}[reduction from the discrete to the continuous problem]
    Let $U_1, \ldots, U_n \simiid U[0,1]$ and let $B_1, \ldots, B_{k}$ be a set of upper bounds in the discrete boundary
crossing probability \eqref{eq:order_statistics_upperbound_problem}.
    Define their cumulative function as $b(t) = \sum_{i=1}^{k} 1(B_i \le t)$, then
    \begin{align}
        \pr{\forall i \in \{1, \ldots, k\} : U_{(i)} \le B_i} = \pr{\forall t: b(t) \le n F_n(t)},
    \end{align}
    where $F_n$ is the empirical CDF of $U_1, \ldots, U_n$.
\end{corollary}
\begin{proof}
    By the construction of $b(t)$, for all $i$, $ B_i = \inf\{t : b(t) > i-1 \} $.
    This coincides with the definition of $B_i$ in Eq. \eqref{eq:t_i_def}.
    Equation \eqref{eq:cor:reduction_to_finite} follows.
\end{proof}
\begin{remark}
For data from a non-uniform distribution $X_i \sim F$, we may transform the variables as $U_i = F(X_i)$.
If the distribution $F$ is continuous then $U_i \sim U[0,1]$, thus we may directly apply the reductions above to the
transformed variables as described in Section~\ref{sec:pvalue}.
However, discontinuous distributions require a more intricate analysis.
For the full details of the reduction in the discontinuous case, see Theorem 1 of \cite{Gleser1985} which extends Corollary~\ref{cor:reduction_to_finite} above.
These results were used by \cite{Dimitrova2020a} to compute the distribution of the Kolmogorov-Smirnov statistic when the
underlying distribution $F$ is discontinuous.
\end{remark}
% \begin{theorem}
%     Let $X_1, \ldots, X_n$ be an i.i.d. sample from a distribution $F$ which may be discontinuous.
%     Let $F_n(x) := \tfrac1n \sum_{i=1}^n \mathbf{1}(X_i \le x)$ be the empirical cumulative distribution.
%     Let $b:\mathbb{R} \to \mathbb{R}$ be a function, then  
%     \begin{align}
%         \pr{\forall x: b(x) \le nF_n(x)}
%         =
%         \pr{\forall i: U_{(i)} \le F(B_i)}
%     \end{align}
%     where  $B_i = \inf\{ x : b(x) > i-1\}$.
% \end{theorem}
% This is a natural extension of Corollary \ref{cor:reduction_to_finite} to discontinuous distributions.
% Furthermore, \cite{Gleser1985} proved similar results for upper boundaries and two-sided boundaries of the empirical CDF.
% Interestingly, the upper boundary function must be right-continuous.

% \begin{theorem}[\cite{Gleser1985}]
%     \begin{align}
%         \pr{\forall x: G_2(x) \le F_n(x) \le G_1(x)}
%         =
%         \pr{\forall i: F(a_i-) \le U_{(i)} \le F(b_i)}
%     \end{align}
%     where $F(x-) = \sup_{z<x} F(z)$, $a_i = \inf \{ x : G_1(x) \ge i/n\}$ and $b_i = \inf\{ x : G_2(x) > (i-1)/n\}$.
% \end{theorem}

\section[Appendix B. Benchmark and implementation details]{Benchmark and implementation details} \label{sec:implementation}

All four methods compared in Figure~\ref{fig:benchmarks} were implemented in  C++,  compiled in \textsf{clang} 11.0.3, and tested on a 2019 Intel Core i7-8569U CPU.
For computing the fast Fourier transform we used the library \texttt{FFTW} 3.3.8 in single-threaded mode \citep{FrigoJohnson2005}.

In the calculation of $Q(i+1,:)$ according to Eq. \eqref{eq:Q_next}, we represent the zero elements implicitly, thus reducing the size of the FFT convolutions from $n+1$ to $n+1-i$.
This optimization already existed in our previous code for computing two-sided non-crossing probabilities \citep{MoscovichNadler2017}.
We added additional optimizations to
the two-sided Poisson-propagation algorithms {\bf KS~(2001)} and {\bf MN~(2017)} that specifically handles consecutive lower
bounds that satisfy $b_{i+1}=b_i$ as a special case (see Eq. \eqref{eq:two_sided}).
This makes the two methods more competitive for the computation of  one-sided boundary crossing probabilities.
This, in addition to several other technical code optimizations and the improvement in processor speed, resulted in an 8-fold
decrease in the running time of {\bf MN\,(2017)} in the one-sided boundary case, compared to our previous benchmark~\citep{MoscovichNadler2017}.

Our proposed algorithm  has a configurable jump size parameter $k$.
The entire range $k \in [\log n, n / \log n]$ gives asymptotically optimal results of $O(n^2)$.
To get a ballpark estimate for the optimal value of $k$ we set $k(x) = x \sqrt{n}$ and minimized the asymptotic runtime 
in Eq.~\eqref{eq:runtime_given_k}. The resulting minimizer is $k=\sqrt{2n}$. However, the setting used in the benchmarks
was $k = \sqrt{n}$ as this was empirically found to be faster.

%\bigskip
%\small
%\settocbibname{References}
%\bibliographystyle{apalike-refs}

%\section*{References} 
\phantomsection
\addcontentsline{toc}{section}{References}
\bibliographystyle{elsarticle-harv}
\bibliography{crossing-probability}

\end{document}